\documentclass[preprintnumbers,11pt,onecolumn]{article}
\pdfoutput=1	
\usepackage{fullpage}
\usepackage{graphicx, amssymb, amsmath, dsfont, amsthm}
\usepackage{subfigure}

\usepackage[tracking=smallcaps]{microtype}	

\usepackage{color}
\definecolor{DarkGray}{rgb}{0.1,0.1,0.5}
\usepackage{url}
\usepackage[colorlinks=true,breaklinks, linkcolor= DarkGray,citecolor= DarkGray,urlcolor= DarkGray]{hyperref}
\newcommand{\eqnref}[1]{\hyperref[#1]{{(\ref*{#1})}}}
\newcommand{\thmref}[1]{\hyperref[#1]{{Theorem~\ref*{#1}}}}
\newcommand{\lemref}[1]{\hyperref[#1]{{Lemma~\ref*{#1}}}}
\newcommand{\corref}[1]{\hyperref[#1]{{Corollary~\ref*{#1}}}}
\newcommand{\defref}[1]{\hyperref[#1]{{Definition~\ref*{#1}}}}
\newcommand{\secref}[1]{\hyperref[#1]{{Section~\ref*{#1}}}}
\newcommand{\figref}[1]{\hyperref[#1]{{Figure~\ref*{#1}}}}
\newcommand{\tabref}[1]{\hyperref[#1]{{Table~\ref*{#1}}}}
\newcommand{\remref}[1]{\hyperref[#1]{{Remark~\ref*{#1}}}}
\newcommand{\appref}[1]{\hyperref[#1]{{Appendix~\ref*{#1}}}}
\newcommand{\claimref}[1]{\hyperref[#1]{{Claim~\ref*{#1}}}}
\newcommand{\exampleref}[1]{\hyperref[#1]{{Example~\ref*{#1}}}}

\newcommand{\comment}[1]{\emph{\color{blue}Comment:\color{black} #1}} 
 
\newlength{\commentslength}
\newcommand{\comments}[1]{
\hspace{-2\parindent}
\addtolength{\commentslength}{-\commentslength}
\addtolength{\commentslength}{\linewidth}
\addtolength{\commentslength}{-\parindent}
\fcolorbox{blue}{white}{\smallskip\begin{minipage}[c]{\commentslength}
\emph{Comments:}\begin{itemize}#1\end{itemize}\end{minipage}}\bigskip
}
\renewcommand{\comment}[1]{}\renewcommand{\comments}[1]{}

\newcommand{\ket}[1]{\left| #1\right\rangle}      
\newcommand{\bra}[1]{\left\langle #1\right|}

\newcommand{\abs}[1]{\left|#1 \right|}			
\newcommand{\N}{\mathbb{N}}

\newtheorem{theorem}{Theorem}[section]
\newtheorem{lemma}[theorem]{Lemma}

\def\beq{\begin{equation}}
\def\eeq{\end{equation}}

\newcommand{\fastmatrix}[1]{\left(\begin{smallmatrix}#1\end{smallmatrix}\right)}

\usepackage{tikz}
 
\usetikzlibrary{arrows,decorations.pathmorphing,backgrounds,positioning,fit,shapes.misc}

\begin{document}

\title{Systematic distillation of composite Fibonacci anyons\\ using one mobile quasiparticle}
\author{Ben W.~Reichardt \\ \small Department of Electrical Engineering, University of Southern California}
\date{}

\maketitle

\begin{abstract}
A topological quantum computer should allow intrinsically fault-tolerant quantum computation, but there remains uncertainty about how such a computer can be implemented.  It is known that topological quantum computation can be implemented with limited quasiparticle braiding capabilities, in fact using only a single mobile quasiparticle, if the system can be properly initialized by measurements.  It is also known that measurements alone suffice without any braiding, provided that the measurement devices can be dynamically created and modified.  We study a model in which both measurement and braiding capabilities are limited.  Given the ability to pull nontrivial Fibonacci anyon pairs from the vacuum with a certain success probability, we show how to simulate universal quantum computation by braiding one quasiparticle and with only one measurement, to read out the result.  The difficulty lies in initializing the system.  We give a systematic construction of a family of braid sequences that initialize to arbitrary accuracy nontrivial composite anyons.  Instead of using the Solovay-Kitaev theorem, the sequences are based on a quantum algorithm for convergent search.  
\end{abstract}

\section{Introduction}

In a topological quantum computer, universal quantum computation can be simulated by braiding quasiparticle excitations around each other on a two-dimensional surface~\cite{Kitaev97anyons}.  Provided the quasiparticles are kept well apart, the computation is intrinsically highly resilient to errors.  Topological quantum computers can be implemented, or simulated, in a variety of ways~\cite{NayakSimonSternFreedmanDasSarma08tqcreview}.  For example, we might find lattice spin systems in which certain necessary local interactions arise naturally~\cite{Kitaev97anyons, Kitaev05hexagon}, or the interactions could be artificially engineered~\cite{LevinWen05stringnet}.  Physical quantum systems that may more directly allow for universal topological quantum computation include fractional quantum Hall systems~\cite{NayakSimonSternFreedmanDasSarma08tqcreview} and topological insulators with topologically nontrivial surfaces~\cite{BondersonDasSarmaFreedmanNayak10insulator, HasanKane10topologicalinsulators}.  

There has been steady progress in fabricating these substrates and then studying their quantum properties to verify conjectured theories.  However, engineering a topological quantum computer based on such a system remains a formidable challenge.  A direction of theoretical research has therefore been to design topological quantum computational schemes that minimize the required resources.  For example, moving quasiparticles appears to be difficult---the technology is speculative~\cite{FreedmanNayakWalker05Ising}.  Simon et al.~\cite{SimonBonesteelFreedmanPetrovicHormozi05onemobile} showed how to implement universal quantum computation by weaving only a single mobile quasiparticle through an array of stationary quasiparticles.  Bonderson, Freedman and Nayak~\cite{BondersonFreedmanNayak08measurementonlytqcprl, BondersonFreedmanNayak08measurementonlytqc} developed a scheme that requires no quasiparticle braiding.  Instead, their approach relies only on certain collective anyon measurements to teleport anyons around each other.  They envision using dynamic deformations of the fractional quantum Hall medium in order to place interferometers around the anyons to be measured.  Limited deformations have been experimentally achieved~\cite{WillettPfeifferWest08fqhe}, but a full implementation of these measurements may be as difficult as braiding quasiparticles.  For example, in the Ising model, it may be difficult to calibrate the interferometers to distinguish trivial charge from charge $\psi$~\cite{BondersonFreedmanNayak08measurementonlytqc}.  

In this paper, we consider the case in which both measurement and quasiparticle braiding operations are difficult.  In a model meant as a compromise between references~\cite{SimonBonesteelFreedmanPetrovicHormozi05onemobile} and~\cite{BondersonFreedmanNayak08measurementonlytqcprl}, we allow for weaving only a single mobile quasiparticle, and also restrict measurement to the fusion channel of only one anyon pair.  This measurement is used exactly once at the very end of the computation, to read out the result.  However, limited measurement capability makes initializing the system a problem.  It disallows the standard approach of preparing anyons: trapping stray quasiparticles, and then measuring interferometrically to check for the presence of a nontrivial anyon charge and break entanglement inside the system.  Therefore our model instead supposes that nontrivial quasiparticle pairs can be created from the vacuum each with a constant probability above zero, independently.  

Quasiparticles that stochastically may or may not be trivial are not directly useful for computation.  To separate this entropy and create effective pure states, we use the method of composite anyon distillation, introduced in~\cite{Koenig09distillation}.  In composite anyon distillation, sketched in \figref{f:models}, a collection of unentangled quasiparticles, each of which has a certain probability independently of being nontrivial, is manipulated in order to create a composite anyon that with high probability is nontrivial.  In further computation, the collection of quasiparticles is then treated as a single entity.  (If the physical model allows it, they can be fused together.)  Our contribution is in giving a scheme for distilling composite anyons by weaving only a single mobile quasiparticle.  Once composite anyons can be created with high probability, they can be manipulated by braiding the single quasiparticle to implement the desired computation following~\cite{SimonBonesteelFreedmanPetrovicHormozi05onemobile}, and finally the outcome of the computation can be read out using a measurement.  

\begin{figure}
\centering
\raisebox{.1in}{
\def\halfcircle#1#2#3#4#5{\draw [#5] (#1,#3) -- (#1,#4) -- (#2,#4) -- (#2,#3);}
\def\starx#1{\node at (#1,.15) {$\star$};}
\def\mylabel#1{\text{(#1)}\!\!\!\!\!\!\!\!\!\!\!\!\!\!\!\!}
\begin{tikzpicture}[scale=1,baseline=-17.5pt]
\halfcircle{-1.5}{-.5}{0}{-.5}{rounded corners=2.6ex}
\starx{-.5};
\halfcircle{-.1}{.9}{0}{-.5}{rounded corners=2.6ex}
\halfcircle{1.3}{2.3}{0}{-.5}{rounded corners=2.6ex,dashed}
\halfcircle{2.7}{3.7}{0}{-.5}{rounded corners=2.6ex}
\halfcircle{4.1}{5.1}{0}{-.5}{rounded corners=2.6ex,dashed}
\draw [thick,loosely dotted] (2.5,.3) -- (2.5,-1.5);
\end{tikzpicture}
$\quad\longrightarrow$
\raisebox{.1in}{
\def\halfcircle#1#2#3#4#5{\draw [#5] (#1,#3) -- (#1,#4) -- (#2,#4) -- (#2,#3);}
\def\starx#1{\node at (#1,.15) {$\star$};}
\def\mylabel#1{\text{(#1)}\!\!\!\!\!\!\!\!\!\!\!\!\!\!\!\!}
\begin{tikzpicture}[scale=1,baseline=-10pt]
\halfcircle{-1.5}{-.5}{0}{-.5}{rounded corners=2.6ex}
\starx{-.5};
\halfcircle{0}{1}{0}{-.5}{rounded corners=2.6ex}
\halfcircle{1.5}{2.5}{0}{-.5}{rounded corners=2.6ex,dashed}
\halfcircle{3}{4}{0}{-.5}{rounded corners=2.6ex}
\halfcircle{4.5}{5.5}{0}{-.5}{rounded corners=2.6ex,dashed}
\draw [rounded corners=2.3ex,dashed] (2,-.5) -- (2,-1) -- (1.25,-1);
\draw [rounded corners=1.5ex] (.5,-.5) -- (.5,-1) -- (1.25,-1) -- (1.25,-1.5) -- (2.75,-1.5) -- (4.25,-1.5) -- (4.25,-1) -- (3.5,-1) -- (3.5,-.5);
\draw [rounded corners=2.3ex,dashed] (5,-.5) -- (5,-1) -- (4.25,-1);
\draw [thick,loosely dotted] (2.75,.3) -- (2.75,-1.5);
\end{tikzpicture}
}}
\caption{Composite anyon distillation: For initializing the anyon system, we allow a basic operation that attempts to pull from the vacuum a nontrivial anyon pair.  It succeeds with probability $p > 0$, but the success or failure (solid or dashed lines, respectively) is not revealed.  By weaving a single, nontrivial mobile quasiparticle (marked $\star$) around the others, we distill across a partition (dotted line) a composite anyon that is nontrivial with high probability.  Entropy is unchanged by this unitary operation, but remains localized in the higher portion of the fusion diagram.} \label{f:models}
\end{figure}
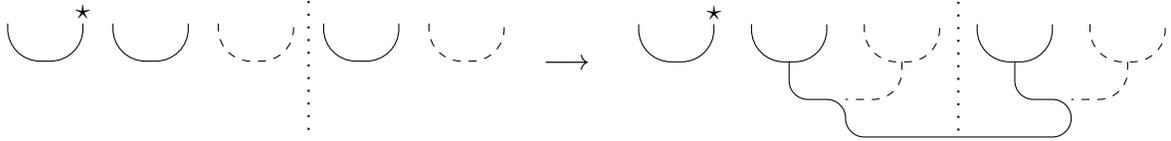

Besides extending K{\"o}nig's distillation scheme to the case of a single mobile quasiparticle, we also improve the efficiency of the scheme.  In particular, in the original model, in which nontrivial anyon pairs can be created with probability $p$, K{\"o}nig's method uses $O\big( \frac 1 {p^2} (\log \frac 1 \epsilon)^{5+\delta}\big)$ physical braids to distill a composite anyon with probability $1-\epsilon$, where $\delta > 0$ can be any constant.  The improved implementation uses only $O\big( \frac 1 {p^2} (\log \frac 1 \epsilon)^3 \big)$ physical braids for the same accuracy.  This improvement is made possible by giving a more efficient braid sequence for achieving the same operations.  

Although the Fibonacci model is universal, it does not allow the exact implementation of every desired operation~\cite{FreedmanWang06Fibexact}.  Instead, desired operations must be approximated by sequences of available gates each corresponding to an elementary quasiparticle braiding operation.  
Braid sequences are generally found by brute-force search, sometimes after using dimensionality-reduction ideas to simplify the search for two-qubit operations~\cite{BonesteelHormoziZikosSimon05compile, HormoziZikosBonesteelSimon07compile, HormoziBonesteelSimon09SU2k, XuWan08compile2, XuWan08compile, XuTaylor10compile}.  After the initial search achieves a certain constant accuracy, an appeal is made to the Solovay-Kitaev theorem~\cite{KitaevShenVyalyi02text, DawsonNielsen05solovaykitaev} to constructively derive arbitrarily accurate braids, with the number of braids required to achieve an error $\epsilon$ scaling as $O\big( (\log \frac 1 \epsilon)^{3+\delta} \big)$, where $\delta > 0$ can be any constant.  There are indications that the Solovay-Kitaev theorem may be too pessimistic for the Fibonacci model~\cite{Mosseri08fib, BurrelloXuMussardoWan09hashing}, but no stronger convergence guarantee is yet known~\cite{HarrowRechtChuang01SolovayKitaev}.  

Instead of relying on search and the Solovay-Kitaev theorem, we give a completely explicit distillation braid sequence for which the error can be analyzed exactly.  A systematic, iterative procedure allows for achieving arbitrarily small error.  The number of braids required to realize an error $\epsilon > 0$ scales as only $O(\log \frac 1 \epsilon)$.  A brute-force search can be used to initialize the procedure, possibly improving the hidden constant in the big-$O$ notation, but such a search is unnecessary.  

This more efficient gate compilation procedure is inspired by Grover's convergent search quantum algorithms~\cite{Grover05search}.  This algorithm applies phases of $\pi/3$ to the source and target states, instead of $\pi$ as in Grover's original search algorithm~\cite{Grover97search}.  Unlike the original algorithm, it does not give a square-root speedup for unstructured database search.  However, when run iteratively, it is a convergent procedure, that converges to the target, instead of rotating past it.  Previous work has used this to design higher-order-accurate composite pulse sequences for qubit control~\cite{ReichardtGrover05}.  Here we apply generalizations of this technique to design braid sequences for controlling Fibonacci anyons.  The construction takes advantage of two properties.  First, we do not need to approximate arbitrary unitaries.  In fact, distillation can be reduced to implementing certain $2 \times 2$ matrices, either the identity or the Pauli $X$ matrix, acting between certain anyon fusion basis states.  Second, braids in the Fibonacci model allow for easily applying phases that are multiples of $\pi/5$.  Although the phase-$\pi/3$ convergent search algorithm does not apply, a phase-$\pi/5$ generalization does.  

Regarding possible applications, there are at least two major caveats to our approach.  First, as in~\cite{Koenig09distillation}, our scheme works for the Fibonacci anyon theory, also known as $SO(3)_3$.  The Fibonacci theory allows for universal quantum computation, and is the non-abelian part of one of the candidates for  the anyon model for fractional quantum Hall liquids at filling fraction $\nu = 12/5$~\cite{RezayiRead09fqhe}.   However, the form of the $\nu = 12/5$ state has not been resolved~\cite{BondersonFeiguinMollerSlingerland09BS}.  The Ising anyon theory, related to $SU(2)_2$, may be more easily accessible in experiments, likely appearing for example in the less fragile $\nu = 5/2$ fractional quantum Hall state and possibly in topological insulators.  Unfortunately, composite anyon distillation is impossible for the Ising model in the plane, since allowed quantum operations can be simulated by Clifford gates~\cite{NayakSimonSternFreedmanDasSarma08tqcreview, BondersonClarkeNayakShtengel09Isinguniversal}.  In any case, the Fibonacci model is worth studying as it is the simplest non-abelian anyon model, and has other possible realizations beyond fractional quantum Hall systems~\cite{LevinWen05stringnet, KoenigKuperbergReichardt10TVcode, CooperWilkinGunn01rotatingbosetqc}.  In the conclusion, we will discuss a way of extending composite anyon distillation to anyon models $SO(3)_k$ where $k+2$ is prime.  

A second caveat is that the operations we allow may not be the most suitable operations for a particular implementation of the Fibonacci theory.  While interferometric measurements have been implemented in fractional quantum Hall systems~\cite{WillettPfeifferWest08fqhe, WillettPfeifferWest09measurement}, braiding and creation of particle pairs from the vacuum have not.  Thus it remains unclear what set of operations will be most experimentally accessible.  The current work may be seen as exploring one alternative.  As experiments advance, we should obtain a better understanding of the advantages and disadvantages of different approaches.  For example, the presence of static stray quasiparticles could pose a problem both for a measurement-based topological quantum computation scheme, since the different interferometry regions must enclose the same sets of quasiparticles~\cite[Sec.~6.4]{BondersonFreedmanNayak08measurementonlytqc}, and also for our scheme, since then regions of anyons that should fuse to the vacuum might not.  In the conclusion, we will discuss a possible solution, should this turn out to be an issue: set up a single interferometer with one trap inside the interferometry region, and then braid a mobile quasiparticle in and out of this trap, and around other traps, in order to measure the anyons at other positions.  

\smallskip
This paper is organized as follows.  \secref{s:systematic} begins by stating the simple matrix identity that is the basis for our systematic construction of higher-order-accurate braid sequences.  The identity is a generalization of Grover's convergent search algorithm.  \secref{s:fibonacci} briefly reviews the parameters of the Fibonacci anyon model.  \secref{s:distillationbybraiding} considers composite anyon distillation in the case that all quasiparticles can be moved for braiding, and presents a completely explicit braid sequence that is more efficient than the sequence given for the same model in~\cite{Koenig09distillation}.  \secref{s:distillationbyweaving} shows how to distill composite anyons in the one-mobile-quasiparticle model.  Unlike the schemes in \secref{s:distillationbybraiding} and~\cite{Koenig09distillation}, the presented method does not use hierarchical distillation, and is thus even simpler in certain ways.  However, the limited quasiparticle mobility also introduces some technical problems.  Finally, \secref{s:conclusion} concludes with a discussion of some extensions and open problems.

\section{Systematic construction of higher-order-accurate \texorpdfstring{$2 \times 2$}{2x2} unitaries} \label{s:systematic}

The explicit braid sequences we will derive are based on Grover's convergent search algorithm, which is a variation of the well-known amplitude amplification algorithm.  

The basic matrix identity behind the amplitude amplification algorithm~\cite{BrassardHoyerTapp98amplitudeamplify, Grover98amplitudeamplify} is that for any $2 \times 2$ unitary matrix $U$, the $(1,1)$ entry of 
\begin{equation}
A(U) = U \fastmatrix{-1&0\\0&1} U^\dagger \fastmatrix{-1&0\\0&1} U
\end{equation}
equals in magnitude $\cos (3 \arccos \abs{U_{1,1}})$.  Think of $U$ as changing basis, from certain states $\ket s, \ket {s^\perp}$ to $\ket t, \ket {t^\perp}$, with $\abs{\bra t U \ket s}$ small.  Then $A(U)$ uses three calls to $U$ or $U^\dagger$ to amplify the coefficient for going from the ``source" $\ket s$ to the ``target" $\ket t$ by about a factor of three: $\abs{\bra t A(U) \ket s} = \abs{\sin(3 \arcsin \abs{\bra t U \ket s})}$.  The probability of measuring the target is increased by about a factor of nine.  Iterating this procedure by implementing $A(A(U))$, $A(A(A(U)))$, etc., results in the well-known square-root speedup, used for example in Grover's unstructured database search algorithm~\cite{Grover97search, Grover02search}.  

One potential problem in amplitude amplification or Grover's search algorithm is that running the procedure for too long results in the output state turning beyond the target.  Grover's construction of a convergent search algorithm addresses this issue~\cite{Grover05search, Hoyer05piby3}.  In convergent search, the basic iteration 
\begin{equation}
U \longmapsto A'(U) = U \fastmatrix{e^{\pi i/3}&0\\0&1} U^\dagger \fastmatrix{e^{\pi i/3}&0\\0&1} U
\end{equation}
adds phases of $\pi/3$ instead of $\pi$ to the source and target.  It satisfies $\abs{\bra {t^\perp} A'(U) \ket s} = \abs{\bra {t^\perp} U \ket s}^3$.  Hence iterating the map results in the $(1,1)$ coefficient converging to $1$ in magnitude.  Although this algorithm does not give a square-root speedup over classical search, because of its coherency it has proved useful as a subroutine in other quantum algorithms~\cite{WocjanAbeyesinghe08, WocjanChiangAbeyesingheNagaj08, PoulinWocjan09sampling}.  The convergence property is also naturally applied to correct systematic control errors~\cite{ReichardtGrover05}.  

The following lemma, remarked in~\cite{ReichardtGrover05}, generalizes the convergent search algorithm to give fifth-order accuracy based on phase rotations by multiples of $\pi/5$ in between five alternating applications of $U$ and $U^\dagger$.  \figref{f:anglesintuition} gives some geometrical intuition for the choice of angles.  

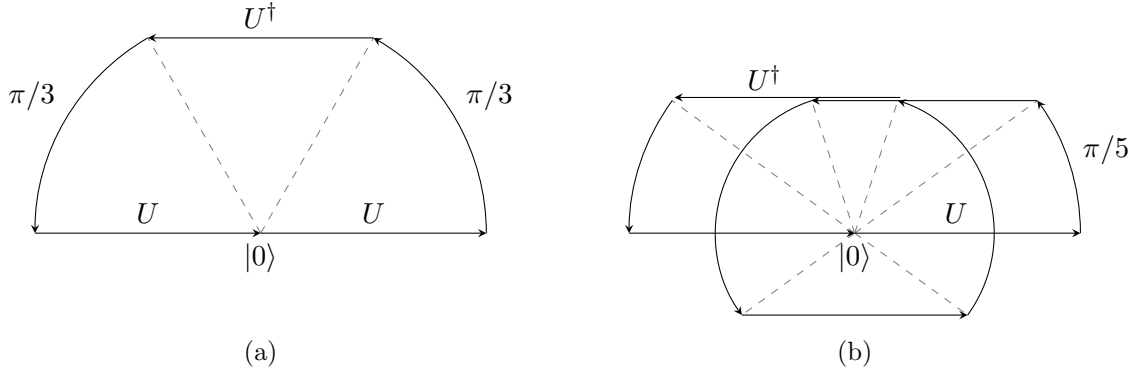
\begin{figure}
\centering
\begin{tikzpicture}[scale=3,baseline=0pt,>=stealth]
\begin{scope}[color=gray]
\draw[dashed] (0,0) -- (60:1);
\draw[dashed] (0,0) -- (120:1);
\end{scope}
\draw[->] (0,0) -- +(1,0);
\draw[->] (1,0) arc(0:60:1);
\draw[->] (60:1) -- +(-1,0);
\draw[->] (120:1) arc(120:180:1);
\draw[->] (-1,0) -- (0,0);
\node [below] at (0,0) {$\ket 0$};
\node [above] at (.5,0) {$U$};
\node [above] at (-.5,0) {$U$};
\node [above] at (0,.866025) {$U^\dagger$};
\node [above right] at (30:1) {$\pi/3$};
\node [above left] at (150:1) {$\pi/3$};
\node [below] at (0,-.43) {\small (a)};
\end{tikzpicture}
$\qquad\quad$
\begin{tikzpicture}[scale=3,baseline=0pt,>=stealth]
\def\onepertau{0.618033989} 
\begin{scope}[color=gray]
\draw[dashed] (0,0) -- (36:1);
\draw[dashed] (0,0) -- (72:\onepertau);
\draw[dashed] (0,0) -- (108:\onepertau);
\draw[dashed] (0,0) -- (144:1);
\draw[dashed] (0,0) -- (-36:\onepertau);
\draw[dashed] (0,0) -- (-144:\onepertau);
\end{scope}
\draw[->] (0,0) -- +(1,0);
\draw[->] (1,0) arc(0:36:1);
\draw[->] (36:1) -- +(-1,0);
\draw[->] (106:\onepertau) arc(106:216:\onepertau);
\draw[->] (216:\onepertau) -- +(1,0); 
\draw[<-] (0,0) -- +(-1,0);
\draw[<-] (-1,0) arc(180:144:1);
\draw[<-] (143:1) -- +(1,0);
\draw[<-] (72:\onepertau) arc(72:-36:\onepertau);
\node [below] at (0,0) {$\ket 0$};
\node [above] at (.45,0) {$U$};
\node [above] at (-.4,.588) {$U^\dagger$};
\node [above right] at (15:1) {$\pi/5$};
\node [below] at (0,-.43) {\small (b)};
\end{tikzpicture}
\caption{Geometrical intuition for convergent search~\cite{ReichardtGrover05}.  These two diagrams show stereographic projections of the Bloch sphere in a small neighborhood of the origin, $\ket 0$.  (a) Beginning at~$\ket 0$, applying the small rotation $U$ (right arrow), a $\pi/3$ rotation about $\ket 0$, $U^\dagger$ (left arrow), another $\pi/3$ rotation, and finally $U$, the state returns to its initial position up to third order.  Part (b) shows a similar geometrical argument for the third-order convergence of the sequence of Eq.~\eqnref{e:Iconverge}, using five applications of $U$ or $U^\dagger$.  By correcting for the curvature of the sphere, one can show that this sequence in fact gives fifth-order convergence (\lemref{t:converge}).} \label{f:anglesintuition}
\end{figure}

\begin{lemma} \label{t:converge}
Let $U \in U(2)$ be a one-qubit unitary, i.e., a $2 \times 2$ unitary matrix in the orthonormal basis $\{ \ket 0, \ket 1 \}$.  Let $\omega = e^{\pi i / 5}$.  
Then 
\begin{equation} \label{e:Iconverge}
\abs{\bra 1 
U \fastmatrix{1&0\\0&\omega} U^\dagger \fastmatrix{1&0\\0&-\omega^{-2}} U \fastmatrix{1&0\\0&-\omega^{-2}} U^\dagger \fastmatrix{1&0\\0&\omega} U
 \ket 0} = \abs{\bra 1 U \ket 0}^5
 \enspace ,
\end{equation}
and furthermore, 
\begin{equation} \label{e:Xconverge}
\abs{\bra 0 
U \fastmatrix{1&0\\0&\omega^{-1}} U^\dagger \fastmatrix{1&0\\0&-\omega^{-2}} U \fastmatrix{1&0\\0&-\omega^2} U^\dagger \fastmatrix{1&0\\0&\omega} U
 \ket 0} = \abs{\bra 0 U \ket 0}^5
 \enspace .
\end{equation}
\end{lemma}

\begin{proof}
Without loss of generality, we may assume that $U = \fastmatrix{\cos \theta & \sin \theta \\ \sin \theta & -\cos \theta}$ for some angle $\theta$.  Indeed, a general $2 \times 2$ unitary matrix can be written $\fastmatrix{c e^{i \theta_{00}} & s e^{i \theta_{01}} \\ s e^{i \theta_{10}} & c e^{i \theta_{11}}}$, where $(c, s) = (\cos\theta, \sin\theta)$ for some $\theta$, and $\theta_{00} - \theta_{01} = \theta_{10} - \theta_{11} + \pi \mod 2\pi$.  Thus we obtain 
\begin{equation}
\fastmatrix{c e^{i \theta_{00}} & s e^{i \theta_{01}} \\ s e^{i \theta_{10}} & c e^{i \theta_{11}}}
= 
\fastmatrix{1 & 0 \\ 0 & e^{i(\theta_{10}-\theta_{00})}}
\fastmatrix{c & s \\ s & -c}
\fastmatrix{e^{i \theta_{00}} & 0 \\ 0 & e^{i \theta_{01}}}
 \enspace .
\end{equation}
Since terms $U$ and $U^\dagger$ alternate in the matrix product on the left-hand side of Eq.~\eqnref{e:Iconverge}, the diagonal matrices to the left and right of $\fastmatrix{c & s \\ s & -c}$ above all cancel out, except for irrelevant phases at the very beginning and very end.  Eq.~\eqnref{e:Iconverge} follows from writing out the matrix product.  Eq.~\eqnref{e:Xconverge} follows from making the substitutions $U \rightarrow X U$ and $\omega \rightarrow \omega^*$ in Eq.~\eqnref{e:Iconverge}.  
\end{proof}

We remark that further generalizations of convergent search are possible.  Tulsi, Grover and Patel~\cite{TulsiGroverPatel05search, GroverTulsiPatel06search} defined higher-order-accurate generalizations with a more complicated structure, using extra ancilla states.  The most straightforward generalization, though, is for $k \in \N$ to consider the product $Q_k U^{(-1)^k} P_k$, where $P_j$ and $Q_j$ are defined inductively by $P_0 = Q_0 = I$ and 
\begin{equation}\begin{split}
P_{j+1} &= \fastmatrix{1&0\\0&(-1)^j \omega^{(-1)^j (j+1)}} U^{(-1)^j} P_j \\
Q_{j+1} &= Q_j U^{(-1)^j} \fastmatrix{1&0\\0&(-1)^j \omega^{(-1)^j (j+1)}}  \enspace ,
\end{split}\end{equation}
with $\omega = e^{\pi i / (2k+1)}$.  This product appears to give accuracy to order $2k+1$, which would be optimal by a polynomial degree argument~\cite{ChakrabortyRadhakrishnanRaghunathan05search}.  

For deriving braid sequences in the Fibonacci anyon model, these further generalizations are unnecessary, because the available phases introduced by a single braid include only multiples of~$\pi/5$.  However, the generalizations may be of interest for compiling braid sequences in anyon models $SO(3)_k$ or $SU(2)_k$ for odd $k > 3$.  Such models make available angles that are multiples of $\pi / (k+2)$.

\section{The Fibonacci anyon model} \label{s:fibonacci}

The Fibonacci anyon model~\cite{FreedmanLarsenWang02Fib, TrebstTroyerWangLudwig09fib, Bonderson07thesis}, also known as $SO(3)_3$, is perhaps the simplest nonabelian anyon model.  It is specified as follows.  

There are two particle types, $0$ the trivial particle type, and $1$ (sometimes called, respectively $1$ and $\tau$, or $I$ and $\epsilon$).  The only nontrivial fusion rule is $1 \otimes 1 = 0 \oplus 1$.  

The braid matrix is trivial, except for braiding two $1$ anyons.  The effect of a counterclockwise exchange of two $1$ anyons depends on their fusion $b \in \{0,1\}$, and is given in the $0, 1$ basis by 
\begin{equation} \label{e:Rdef}
R := \fastmatrix{e^{-4 \pi i/5} & 0 \\ 0 & e^{3 \pi i/5}}
 \enspace .
\end{equation}
Diagrammatically, we have 
\begin{equation}
\def\paircoordinates{\coordinate (A) at (-.3,0);\coordinate (B) at (.3,0);\coordinate (AB) at (0,-.4);\coordinate (C) at (0,-.75);}
\begin{tikzpicture}[scale=1,baseline=-13pt]
\paircoordinates
\draw [out=-30,in=30,looseness=1.8] (A) to (AB);
\draw [draw=white,double=black,double distance=.4pt,very thick,out=-150,in=150,looseness=1.8] (B) to (AB);
\draw (AB) -- (C);
\node at (.2,-.6) {$b$};
\end{tikzpicture}
=
\bra b R \ket b
\begin{tikzpicture}[scale=1,baseline=-13pt]
\paircoordinates
\draw [out=-80,in=150,looseness=1] (A) to (AB);
\draw [out=-100,in=30,looseness=1] (B) to (AB);
\draw (AB) -- (C);
\node at (.2,-.6) {$b$};
\end{tikzpicture}
\end{equation}
where the convention is that unlabeled edges are $1$ anyons.  

The $F$ matrix, describing the associativity of fusion, is trivial except for the case of three $1$ anyons fusing into a $1$ anyon.  In this case, it relates the bases 
\begin{equation}
\def\threecoordinates{\coordinate (A) at (-.5,0);\coordinate (B) at (0,0);\coordinate (C) at (.5,0);\
\coordinate (AB) at (-.25,-.25);\coordinate (BC) at (.25,-.25);\
\coordinate (ABC) at (0,-.5);\coordinate (D) at (0,-.75);}
\begin{tikzpicture}[scale=1,baseline=-13pt]
\threecoordinates
\draw (A) -- (AB) -- (B);
\draw (AB) -- (ABC) -- (C);
\draw (ABC) -- (D);
\node at (-.3,-.5) {$b$};
\end{tikzpicture}
=
\sum_{b' \in \{0,1\}}
\bra {b'} F \ket b
\begin{tikzpicture}[scale=1,baseline=-13pt]
\threecoordinates
\draw (B) -- (BC) -- (C);
\draw (A) -- (ABC) -- (BC);
\draw (ABC) -- (D);
\node at (.3,-.5) {$b'$};
\end{tikzpicture}
\end{equation}
where 
\begin{equation}
F := \fastmatrix{\frac 1 \tau & \frac 1 {\sqrt \tau} \\ \frac 1 {\sqrt \tau} & -\frac 1 \tau}
 \enspace ,
\end{equation}
and $\tau := (1 + \sqrt 5)/2$ is the golden ratio.  
The Frobenius-Schur indicators for the two particle types are both $\chi_0 = \chi_1 = 1$.

\section{Systematic distillation of composite Fibonacci anyons} \label{s:distillationbybraiding}

\newcommand{\Fib}{\mathrm{Fib}}
\def\twopaircoordinates{\coordinate (A) at (-1.5,0);\coordinate (B) at (-.5,0);\coordinate (C) at (.5,0);\coordinate (D) at (1.5,0);\coordinate (AB) at (-1,-.5);\coordinate (CD) at (1,-.5);\coordinate (BC) at (0, -.5);\coordinate (AD) at (0,-1);}
\def\AB{\draw [out=-90,in=180] (A) to (AB) to [out=0,in=-90] (B);}
\def\CD{\draw [out=-90,in=180] (C) to (CD) to [out=0,in=-90] (D);}
\def\BC{\draw [out=-90,in=180] (B) to (BC) to [out=0,in=-90] (C);}
\def\AD{\draw [out=-90,in=180] (A) to (AD) to [out=0,in=-90] (D);}
\def\ABCD{\draw [-] (AB) to [bend right=90,looseness=1] (CD);}
\def\ADBC{\draw [-] (AD) -- (BC);}

\subsection{Creating a nontrivial composite anyon pair from two pairs of anyons, at least one of which is nontrivial}

Consider two pairs of anyons, each pair fusing to $0$.  We aim to apply a braid sequence that creates a composite $1$ anyon crossing from the left pair to the right pair, provided that at least one of the pairs is initially nontrivial.  That is, we want to implement the map 
\begin{equation} \label{e:distillation}
\begin{tikzpicture}[scale=1,baseline=-15pt]
\twopaircoordinates \AB \CD
\node at (-1.7,-.3) {$a_1$};
\node at (.3,-.3) {$a_2$};
\end{tikzpicture}
\quad
\longmapsto
\quad
\begin{cases}
\quad
\tikz[baseline=-15pt]{\twopaircoordinates \AB \CD \ABCD}
&
\text{if $a_1 = a_2 = 1$} \\
\quad
\tikz[baseline=-15pt]{\twopaircoordinates \AB \CD \ABCD
\node at (-1.2,-.2) {$a_1$};
\node at (-.3,-.2) {$a_2$};
\node at (.8,-.2) {$a_1$};
\node at (1.7,-.2) {$a_2$};
\node at (0,-.8) {$a_1+a_2$};
}
&
\text{otherwise}
\end{cases}
\end{equation}
up to phases depending on $a_1$ and $a_2$.  It is easy to generate the second part of the map, for the case $(a_1, a_2) \neq (1, 1)$; simply swap the middle two particles.  However, the case $a_1 = a_2 = 1$ is more complicated, and we will achieve it only in the limit.  

The feasibility of achieving Eq.~\eqnref{e:distillation} asymptotically is argued in~\cite{Koenig09distillation}.  It can be seen as follows.  An arbitrary state of four $1$ particles fusing to the vacuum can be written as a superposition of diagrams of the form 
\begin{equation} \label{e:twopairbasis1}
\begin{tikzpicture}[scale=1,baseline=-10pt]
\twopaircoordinates
\AB \CD \ABCD
\node at (0,-.8) {$b$};
\end{tikzpicture}
\end{equation}
with $b \in \{0,1\}$.  Recall that a pure braid is one in which each of the four particles returns to its initial position.  It is known that when $a_1 = a_2 = 1$, the pure braids generate, up to global phases, a dense subgroup of the $2 \times 2$ unitaries $U(2)$, acting on the qubit $b$~\cite{SimonBonesteelFreedmanPetrovicHormozi05onemobile, FreedmanLarsenWang02Fib, FLW_dense}.  Therefore, for any $\epsilon > 0$, we can find a pure braid that implements Eq.~\eqnref{e:distillation} for the case $a_1 = a_2 = 1$ except for a swap of the middle two particles, up to error $\epsilon$ and up to phases.  Being pure, this braid can apply only a phase for the cases $(a_1, a_2) \neq (1, 1)$.  Finally, swapping particles two and three fixes the case $(a_1, a_2) = (1, 1)$, and also guarantees that there will be a composite $1$ particle crossing from left to right in the cases $(1, 0)$ and $(0, 1)$.  While this is only an existence argument, the Solovay-Kitaev theorem~\cite{KitaevShenVyalyi02text, DawsonNielsen05solovaykitaev} gives an algorithm for constructing the desired braid sequence.  For an error $\epsilon$ and any constant $\delta > 0$, the length of the braid sequence will be $O\big( (\log \frac 1 \epsilon)^{3+\delta} \big)$.  

We will argue that the same accuracy can be achieved using only $O(\log \frac 1 \epsilon)$ braid moves, and will specify the braid moves explicitly.  

Start by considering the case of two $1$ anyon pairs, $a_1 = a_2 = 1$.  In the basis of Eq.~\eqnref{e:twopairbasis1}, the effect of a counterclockwise braid of the first two anyons is given by $R$, while the effect of a braid of the middle two anyons is given by $S := F R F$.  Indeed, $F$ changes into the basis 
\begin{equation}
\begin{tikzpicture}[scale=1,baseline=-10pt]
\twopaircoordinates
\AD \BC \ADBC
\node at (.2,-.7) {$b$};
\end{tikzpicture}
\end{equation}
in which $R$ is a braid of the middle two anyons, and then $F^\dagger = F$ returns to the original basis.  

Let $M_0 = S$, and define braid sequences $M_1, M_2, \ldots$ inductively by 
\begin{equation}
M_j = M_{j-1} R^{-1} M_{j-1}^\dagger R^3 M_{j-1} R^{-3} M_{j-1}^\dagger R M_{j-1}
 \enspace .
\end{equation}
By Eq.~\eqnref{e:Xconverge}, 
\begin{equation}
\abs{\bra 0 M_j \ket 0} = \abs{\bra 0 M_{j-1} \ket 0}^5 = \cdots = \abs{\bra 0 M_0 \ket 0}^{5^j}
 \enspace .
\end{equation}

A calculation shows that $\bra 1 S \ket 0 \neq 0$.  Therefore the error $\abs{\bra 0 M_j \ket 0}$ in the sequence $M_j$ converges to $0$ doubly exponentially fast in $j$.  To be precise, since $\bra 0 S \ket 0 = e^{4 \pi i/5}/\tau$, $\abs{\bra 0 M_j \ket 0} = 1/\tau^{5^j}$.  On the other hand, the length $\ell_j$ of the $j$th braid sequence grows exponentially.  It satisfies  
\begin{equation}\begin{split}
\ell_0 &= 1 \\
\ell_j &= 5 \ell_{j-1} + 8 = 3 \cdot 5^j - 2
 \enspace .
\end{split}\end{equation}
Overall, therefore the error drops at a rate exponential in the length of the braid sequence; achieving error $\epsilon$ requires $O(\log \frac 1 \epsilon)$ braids.  Note that that this convergence rate is faster than that generically guaranteed by the Solovay-Kitaev theorem.  

However, so far we have only considered the case $a_1 = a_2 = 1$ in Eq.~\eqnref{e:distillation}.  We need to verify that the sequence also works in the other cases, i.e., $(a_1, a_2) \in \{ (0,1), (1,0) \}$.  (If initially both anyon pairs are trivial, then obviously the braid sequence has no effect.)  Let $\sigma_0 \in S_4$ be the permutation~$(23)$, and inductively define $\sigma_j$ as the four-particle permutation implemented by $M_j$: 
\begin{equation}
\sigma_j = \sigma_{j-1} (12) \sigma_{j-1}^{-1} (12) \sigma_{j-1} (12) \sigma_{j-1}^{-1} (12) \sigma_{j-1}
 \enspace .
\end{equation}
This recursion is periodic, and the solution alternates between the swaps $(23)$ and $(13)$: 
\begin{equation}
\sigma_j = \begin{cases} (23) & \text{for $j$ even} \\ (13) & \text{for $j$ odd$\enspace .$} \end{cases}
\end{equation}
With either of these swaps, nontrivial particles end up on opposite sides of the left/right partition, and hence the braid sequences always create a composite $1$ anyon, satisfying Eq.~\eqnref{e:distillation} exactly, up to a phase.

\subsection{Hierarchical recursion to create a composite anyon with high probability} \label{s:asymptoticsallmobile}

Now assume that each pair of quasiparticles begins in an independent mixture $(\Pr[0], \Pr[1])$ of $0$ and $1$ anyons, with $\Pr[1] \geq p$.  Since we have satisfied Eq.~\eqnref{e:distillation} in all four cases, up to error $\epsilon$ in the $a_1 = a_2 = 1$ case, we find that the probability that the braid sequence generates a composite $1$ anyon is at least 
\begin{equation}
2p(1-p) + p^2(1-\epsilon) = 1 - (1-p)^2 - \epsilon p^2
 \enspace .
\end{equation}

As in~\cite{Koenig09distillation}, we can repeat the entire procedure on pairs of composite anyons.  That is, start with four pairs of quasiparticles, apply the above braid sequences to the first two pairs and the last two pairs, and then apply the same braids to the composite anyon pairs.  Iterate this procedure.  A composite anyon will be created provided that at least one of the underlying quasiparticle pairs is nontrivial.  Therefore, if our aim is to create a composite $1$ anyon except with probability $\epsilon$, it is necessary and sufficient to use $n = \Theta(\frac 1 p \log \frac 1 \epsilon)$ underlying quasiparticle pairs.  

The total number of braid operations on physical or composite anyons is $O(n \log \frac n \epsilon)$, i.e., $O(\log \frac n \epsilon)$ braids at each level, to satisfy Eq.~\eqnref{e:distillation} up to error $\epsilon / n$, times $\frac n 2 + \frac n 4 + \frac n 8 + \cdots + 1$, as there are $n/2^k$ composite pairs at iteration level $k$.  However, implementing a single braid of two composite anyons requires multiple physical braid operations, quadratic in the number of physical quasiparticles comprising the composite anyons.  Expanding out the composite anyon braids, the total number of physical braid operations is $O\big( \sum_k \frac n {2^k} (2^{k-1})^2 \log \frac n \epsilon \big) = O(n^2 \log \frac n \epsilon)$.  This simplifies to $O\big( \frac 1 {p^2} (\log \frac 1 \epsilon)^3 \big)$, provided that $p = \Omega(\epsilon \log \frac 1 \epsilon)$---in fact, in applications $p$ is typically a constant, while $\epsilon$ is polynomially small.  Under the same condition, the distillation scheme of~\cite{Koenig09distillation} requires, for any $\delta > 0$, $O\big(n^2 (\log \frac 1 \epsilon)^{3+\delta} \big) = O\big( \frac 1 {p^2} (\log \frac 1 \epsilon)^{5+\delta}\big)$ physical braids for the same accuracy.  If anyon fusion is allowed, by bringing particles together, the respective complexities of the two schemes are $O\big( \frac 1 p (\log \frac 1 \epsilon)^2 \big)$ and $O\big( \frac 1 p (\log \frac 1 \epsilon)^{4+\delta} \big)$.

\section{Systematic distillation of composite Fibonacci anyons using one mobile quasiparticle} \label{s:distillationbyweaving}

Simon et al.~\cite{SimonBonesteelFreedmanPetrovicHormozi05onemobile} considered the question of whether it is possible to achieve universal quantum computation using Fibonacci anyons under the assumption that only one of the physical quasiparticles can be moved.  This may be a reasonable experimental constraint.  

In this section, we study whether composite anyons can be distilled with high probability if there is only one mobile quasiparticle.  This is clearly impossible if the mobile quasiparticle itself is trivial.  Therefore let us study the case where the mobile quasiparticle is promised to be nontrivial.  

Assume that we are given the following two operations, implemented up to phases to arbitrary accuracy by moving only the mobile quasiparticle marked $\star$: 
\begin{align} \label{e:addition}
\tikz[baseline=-10pt]{\twopaircoordinates \AD \BC \node at (1.5,.15) {$\star$};}
\quad\longmapsto\quad
\tikz[baseline=-10pt]{\twopaircoordinates \AD \BC \ADBC \node at (1.5,.15) {$\star$};}
\\ \label{e:integration}
\tikz[baseline=-10pt]{\twopaircoordinates \AB \CD \ABCD \node at (1.5,.15) {$\star$};}
\quad\longmapsto\quad
\tikz[baseline=-10pt]{\twopaircoordinates \AD \BC \ADBC \node at (1.5,.15) {$\star$};}
\end{align}

Distillation can then be achieved as follows.  Prepare $2 n$ anyon pairs, each pair fusing to the vacuum.   Assume that each pair of anyons begins in an independent mixture $(\Pr[0], \Pr[1])$ of $0$ and $1$ anyons, with $\Pr[1] \geq p$.  Consider the case that there is at least one nontrivial anyon pair among the first $n$ pairs, and another nontrivial pair among the last $n$.  This occurs with probability at least $(1 - (1-p)^n)^2$, which for $n = m/p$ is at least $1 - 2/e^m$.  

As sketched in \figref{f:onemobiledistillation}, now add using Eq.~\eqnref{e:addition} all the nontrivial anyon pairs on the left, integrating them using Eq.~\eqnref{e:integration} pairwise.  Then do the same for the right, resulting in a state like the one shown in \figref{f:onemobiledistillation}(d).  Note that these steps do not require knowing which anyon pairs are trivial or nontrivial; if a pair is trivial, then the braids through and around it have no effect.  Next, integrate the edge from the left with that from the right with Eq.~\eqnref{e:integration}, and apply the inverse braid sequence of Eq.~\eqnref{e:addition}.  Overall, provided that there is initially at least one nontrivial pair on both sides, this procedure results in the creation of a nontrivial composite anyon across the left/right partition, and restores the mobile quasiparticle to its initial position, unentangled with the rest of the system, shown in \figref{f:onemobiledistillation}(f).  

\begin{figure}
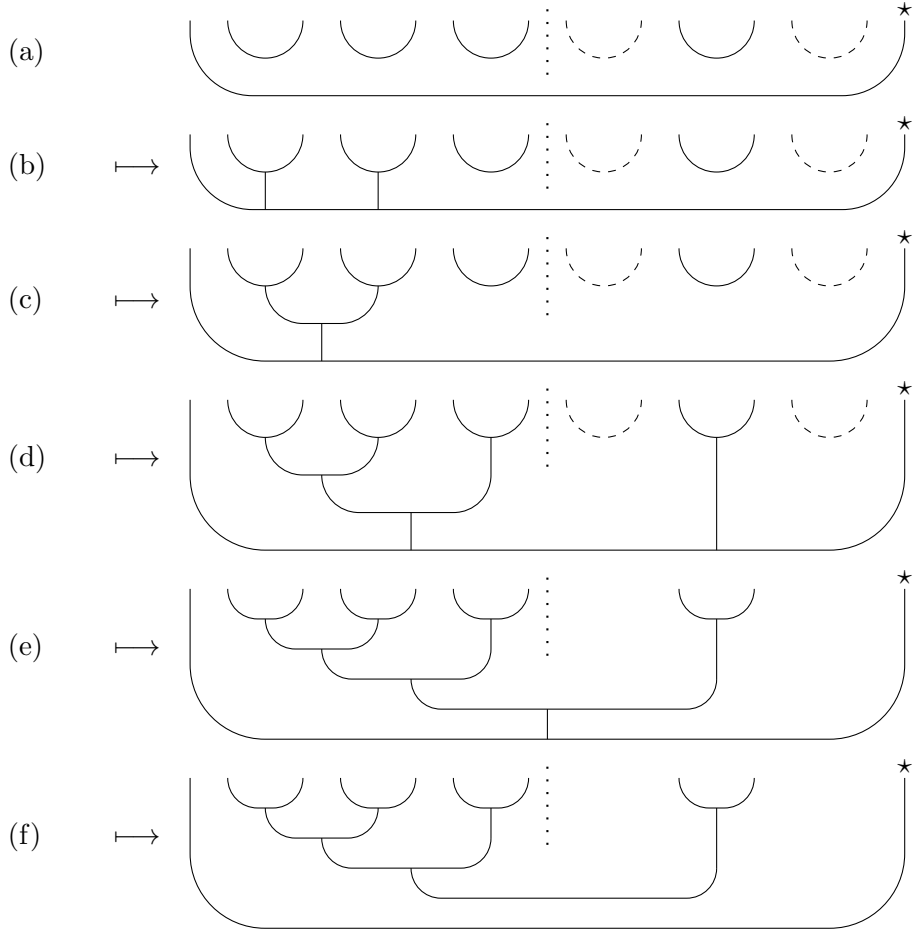

\def\halfcircle#1#2#3#4#5{\draw [#5] (#1,#3) -- (#1,#4) -- (#2,#4) -- (#2,#3);}
\def\starx#1{\node at (#1,.30) {$\star$};}
\def\mylabel#1{\text{(#1)}\!\!\!\!\!\!\!\!\!\!\!\!\!\!\!\!}
\begin{align*}
&\mylabel{a}&
&\tikz[baseline=-15pt,scale=.5]{
\halfcircle{-9.5}{9.5}{0}{-2}{rounded corners=5ex}
\halfcircle{-8.5}{-6.5}{0}{-1}{rounded corners=3ex}
\halfcircle{-5.5}{-3.5}{0}{-1}{rounded corners=3ex}
\halfcircle{-2.5}{-.5}{0}{-1}{rounded corners=3ex}
\halfcircle{.5}{2.5}{0}{-1}{rounded corners=3ex,dashed}
\halfcircle{3.5}{5.5}{0}{-1}{rounded corners=3ex}
\halfcircle{6.5}{8.5}{0}{-1}{rounded corners=3ex,dashed}
\starx{9.5};
\draw [thick,loosely dotted] (0,.3) -- (0,-1.5);
}\\
&\mylabel{b}&\longmapsto\quad
&\tikz[baseline=-15pt,scale=.5]{
\halfcircle{-9.5}{9.5}{0}{-2}{rounded corners=5ex}
\halfcircle{-8.5}{-6.5}{0}{-1}{rounded corners=3ex}
\halfcircle{-5.5}{-3.5}{0}{-1}{rounded corners=3ex}
\halfcircle{-2.5}{-.5}{0}{-1}{rounded corners=3ex}
\halfcircle{.5}{2.5}{0}{-1}{rounded corners=3ex,dashed}
\halfcircle{3.5}{5.5}{0}{-1}{rounded corners=3ex}
\halfcircle{6.5}{8.5}{0}{-1}{rounded corners=3ex,dashed}
\draw (-7.5,-1) -- (-7.5,-2); \draw (-4.5,-1) -- (-4.5,-2);
\starx{9.5};
\draw [thick,loosely dotted] (0,.3) -- (0,-1.5);
}\\
&\mylabel{c}&\longmapsto\quad
&\tikz[baseline=-22.5pt,scale=.5]{
\halfcircle{-9.5}{9.5}{0}{-3}{rounded corners=6ex}
\halfcircle{-8.5}{-6.5}{0}{-1}{rounded corners=3ex}
\halfcircle{-5.5}{-3.5}{0}{-1}{rounded corners=3ex}
\halfcircle{-2.5}{-.5}{0}{-1}{rounded corners=3ex}
\halfcircle{.5}{2.5}{0}{-1}{rounded corners=3ex,dashed}
\halfcircle{3.5}{5.5}{0}{-1}{rounded corners=3ex}
\halfcircle{6.5}{8.5}{0}{-1}{rounded corners=3ex,dashed}
\halfcircle{-7.5}{-4.5}{-1}{-2}{rounded corners=3ex} \draw (-6,-2) -- (-6,-3);
\starx{9.5};
\draw [thick,loosely dotted] (0,.3) -- (0,-2);
}\\
&\mylabel{d}&\longmapsto\quad
&\tikz[baseline=-25pt,scale=.5]{
\halfcircle{-9.5}{9.5}{0}{-4}{rounded corners=6ex}
\halfcircle{-8.5}{-6.5}{0}{-1}{rounded corners=3ex}
\halfcircle{-5.5}{-3.5}{0}{-1}{rounded corners=3ex}
\halfcircle{-2.5}{-.5}{0}{-1}{rounded corners=3ex}
\halfcircle{.5}{2.5}{0}{-1}{rounded corners=3ex,dashed}
\halfcircle{3.5}{5.5}{0}{-1}{rounded corners=3ex}
\halfcircle{6.5}{8.5}{0}{-1}{rounded corners=3ex,dashed}
\halfcircle{-7.5}{-4.5}{-1}{-2}{rounded corners=3ex} 
\halfcircle{-6}{-1.5}{-2}{-3}{rounded corners=3ex} \draw (-1.5,-1) -- (-1.5,-2);
\draw (-3.625,-3) -- (-3.625,-4);
\draw (4.5,-1) -- (4.5,-4);
\starx{9.5};
\draw [thick,loosely dotted] (0,.3) -- (0,-2);
}\\
&\mylabel{e}&\longmapsto\quad
&\tikz[baseline=-25pt,scale=.5]{
\halfcircle{-9.5}{9.5}{0}{-4}{rounded corners=6ex}
\halfcircle{-8.5}{-6.5}{0}{-.8}{rounded corners=2.4ex}
\halfcircle{-5.5}{-3.5}{0}{-.8}{rounded corners=2.4ex}
\halfcircle{-2.5}{-.5}{0}{-.8}{rounded corners=2.4ex}
\halfcircle{3.5}{5.5}{0}{-.8}{rounded corners=2.4ex}
\halfcircle{-7.5}{-4.5}{-.8}{-1.6}{rounded corners=2.4ex} 
\halfcircle{-6}{-1.5}{-1.6}{-2.4}{rounded corners=2.4ex} \draw (-1.5,-.8) -- (-1.5,-1.6);
\halfcircle{-3.625}{4.5}{-2.4}{-3.2}{rounded corners=2.4ex}
\draw (4.5,-.8) -- (4.5,-2.4);
\draw (0,-3.2) -- (0,-4);
\starx{9.5};
\draw [thick,loosely dotted] (0,.3) -- (0,-2);
}\\
&\mylabel{f}&\longmapsto\quad
&\tikz[baseline=-25pt,scale=.5]{
\halfcircle{-9.5}{9.5}{0}{-4}{rounded corners=6ex}
\halfcircle{-8.5}{-6.5}{0}{-.8}{rounded corners=2.4ex}
\halfcircle{-5.5}{-3.5}{0}{-.8}{rounded corners=2.4ex}
\halfcircle{-2.5}{-.5}{0}{-.8}{rounded corners=2.4ex}
\halfcircle{3.5}{5.5}{0}{-.8}{rounded corners=2.4ex}
\halfcircle{-7.5}{-4.5}{-.8}{-1.6}{rounded corners=2.4ex} 
\halfcircle{-6}{-1.5}{-1.6}{-2.4}{rounded corners=2.4ex} \draw (-1.5,-.8) -- (-1.5,-1.6);
\halfcircle{-3.625}{4.5}{-2.4}{-3.2}{rounded corners=2.4ex}
\draw (4.5,-.8) -- (4.5,-2.4);
\starx{9.5};
\draw [thick,loosely dotted] (0,.3) -- (0,-2);
}
\end{align*}
\caption{The steps for distilling a composite anyon using a single mobile quasiparticle.  (a) Begin with an equal number of prepared anyon pairs to the left and right of a dividing line (dotted).  Some pairs may be trivial (dashed).  However, the mobile quasiparticle, marked $\star$, is promised to be nontrivial.  (b) Begin by adding anyon pairs one at a time using Eq.~\eqnref{e:addition}.  (c) After adding each pair, integrate it with the previously added pairs using Eq.~\eqnref{e:integration}.  (d) Continue for all the pairs on the left side.  Then, separately, add and integrate all the pairs on the right side.  Braids around a trivial particle have no effect.  (e) Integrate once more across the left/right partition.  (f) Finally, apply the inverse of Eq.~\eqnref{e:addition} to disentangle the mobile quasiparticle from the others.  Provided there was at least one nontrivial anyon pair on both sides, this results in a nontrivial composite anyon crossing the partition.}  \label{f:onemobiledistillation}
\end{figure}

It remains to show how to implement Eqs.~\eqnref{e:addition} and~\eqnref{e:integration}.  By density for pure braids, we can achieve both of these maps asymptotically, and the Solovay-Kitaev theorem gives an algorithm for constructing better and better approximations.  We will give a systematic construction that converges more rapidly than the guarantee provided by the Solovay-Kitaev theorem.  

For a $2 \times 2$ matrix $M_0$ (to be determined), define matrix products $M_1, M_2, \ldots$ inductively by 
\begin{equation} \label{e:Xsequence}
M_j = M_{j-1} R^{-1} M_{j-1}^\dagger R^3 M_{j-1} R^{-3} M_{j-1}^\dagger R M_{j-1}
 \enspace .
\end{equation}
By Eq.~\eqnref{e:Xconverge}, 
\begin{equation}
\abs{\bra 0 M_j \ket 0} = \abs{\bra 0 M_{j-1} \ket 0}^5 = \cdots = \abs{\bra 0 M_0 \ket 0}^{5^j}
 \enspace .
\end{equation}

For a $2 \times 2$ matrix $N_0$ (to be determined), define matrix products $N_1, N_2, \ldots$ inductively by 
\begin{equation} \label{e:Isequence}
N_j = N_{j-1} R N_{j-1}^\dagger R^3 N_{j-1} R^3 N_{j-1}^\dagger R N_{j-1}
 \enspace .
\end{equation}
By Eq.~\eqnref{e:Iconverge}, 
\begin{equation}
\abs{\bra 1 N_j \ket 0} = \abs{\bra 1 N_{j-1} \ket 0}^5 = \cdots = \abs{\bra 1 N_0 \ket 0}^{5^j}
 \enspace .
\end{equation}

Thus the matrices $M_j$ converge doubly exponentially fast to $\fastmatrix{0&1\\1&0}$, up to phases on the two basis states, while $N_j$ converges doubly exponentially fast to the identity, up to phases.  These are the matrices needed for Eqs.~\eqnref{e:addition} and~\eqnref{e:integration}, respectively, acting from the middle edge of the left diagram, $0$ or $1$, to the middle edge of the right diagram.  

There are two problems.  First, these matrices need to be implemented using braids of the mobile quasiparticle.  Second, we must ensure that the braid sequence ends up in the basis on the right-hand side of Eqs.~\eqnref{e:addition} and~\eqnref{e:integration}, with the mobile quasiparticle in the right-most position.  

\def\smallcoords{\def\Ax{-.9} \def\Bx{-.3} \def\Cx{.3} \def\Dx{.9} \def\y{-.8} \def\yy{-.4} \def\starA{\node at (\Ax,.15) {$\star$};} \def\starB{\node at (\Bx,.15) {$\star$};} \def\starC{\node at (\Cx,.15) {$\star$};} \def\starD{\node at (\Dx,.15) {$\star$};} }
\def\smallbasisone{\smallcoords\
\draw (\Ax,0) -- (\Ax,\y) -- (\Bx,\y); \draw [very thick] (\Bx,\y) -- (\Cx,\y); \draw (\Cx,\y) -- (\Dx,\y) -- (\Dx,0);\
\draw (\Bx,0) -- (\Bx,\y); \draw (\Cx,0) -- (\Cx,\y);}
\def\smallbasistwo{\smallcoords\
\draw (\Ax,0) -- (\Ax,\y) -- (\Dx,\y) -- (\Dx,0);\
\draw (\Bx,0) -- (\Bx,\yy) -- (\Cx,\yy) -- (\Cx,0); \draw [very thick] (0,\yy) -- (0,\y);}

\def\shorttikz#1{\tikz[baseline=-10pt,rounded corners=2ex]{#1}}

The first problem is straightforward to solve given braid implementations of $M_0$ and $N_0$.  Assume that we are currently in one of the fusion bases 
\begin{equation} \label{e:twobases}
\shorttikz{\smallbasisone} \quad\text{or}\quad \shorttikz{\smallbasistwo} \enspace .
\end{equation}
In each diagram, the middle, bold edge can be either $0$ or $1$, i.e., absent or present.  Then regardless of the position of the mobile quasiparticle, an application of the $2 \times 2$ matrix $R$ can be implemented on the middle edge with a certain braid.  In the following four cases, a counterclockwise exchange of the mobile quasiparticle with its nearest neighbor in the fusion diagram implements $R$.  
\begin{equation}\begin{split}
\shorttikz{\smallbasisone \starC} &\quad\longleftrightarrow\quad \shorttikz{\smallbasisone \starD} \\
\shorttikz{\smallbasistwo \starB} &\quad\longleftrightarrow\quad \shorttikz{\smallbasistwo \starC}
\end{split}\end{equation}
In the last two cases, a counterclockwise braid of the mobile quasiparticle about the middle edge 
\begin{equation}
\shorttikz{\smallbasisone \starB}\quad\longleftrightarrow\quad \shorttikz{\smallbasistwo \starD}
\end{equation}
implements the $2 \times 2$ matrix $\fastmatrix{1&0\\0&e^{3\pi i/5}}$ on the middle edge.  The inverse, clockwise braid therefore implements $e^{4 \pi i/5} R$, i.e., $R$ up to an irrelevant global phase.  It is important to notice that none of these operations require weaving the mobile quasiparticle around the leftmost quasiparticle; the leftmost quasiparticle can be far away or even part of the environment, making such a weave expensive or impossible.  

The second problem, though, is that some applications of $R$ change the bases between the two possibilities in Eq.~\eqnref{e:twobases}, and every application changes the position of the mobile quasiparticle.  This means that after implementing $M_j$ or $N_j$ from Eqs.~\eqnref{e:Xsequence} and~\eqnref{e:Isequence}, the basis and the position of the mobile quasiparticle might be incorrect.  Notice, moreover, that in Eqs.~\eqnref{e:Xsequence} and~\eqnref{e:Isequence}, all powers of $R$ are odd, meaning that intermediate braids are certainly not pure.  From Eq.~\eqnref{e:Rdef}, $R$ applies a relative phase between $0$ and $1$ of $e^{-3\pi i/5}$, which is a primitive tenth root of unity.  $R^2$ on the other hand applies a relative phase that is a fifth root of unity, so even powers of $R$ allow a strictly smaller set of relative phases to be applied than do odd powers.  It does not appear to be possible to get the same convergence speed using even powers of the $R$ matrix.  

Note from Eqs.~\eqnref{e:Xsequence} and~\eqnref{e:Isequence} that if $M_0$ and $N_0$ each consist of alternating applications of $F$ and odd powers of $R$, beginning and ending with $F$, then $M_j$ and $N_j$ will have the same form for all~$j$.  Moreover, if $M_0$ and $N_0$ include $m_0$ and $n_0$ $F$ terms, respectively, then $M_j$ and $N_j$ include $5^j m_0$ and $5^j n_0$ $F$ terms, respectively.  

Consider an alternating matrix product $F \ldots F R^{\alpha_2} F R^{\alpha_1} F$, with the $\alpha_i$ odd integers.  By implementing this product on the middle edge by braiding the mobile quasiparticle, the initial basis $\raisebox{.3em}{\scalebox{.5}{\shorttikz{\smallbasistwo \starD}}}$ transforms as 
\begin{equation}
\begin{tabular}{l c l c l}
\shorttikz{\smallbasistwo \starD}
&$\overset{F}{\longrightarrow}$&
\shorttikz{\smallbasisone \starD}
&$\overset{R^{\alpha_1}}{\longrightarrow}$&
\shorttikz{\smallbasisone \starC} \\
&$\overset{F}{\longrightarrow}$&
\shorttikz{\smallbasistwo \starC}
&$\overset{R^{\alpha_2}}{\longrightarrow}$&
\shorttikz{\smallbasistwo \starB} \\
&$\overset{F}{\longrightarrow}$&
\shorttikz{\smallbasisone \starB}
&$\overset{R^{\alpha_3}}{\longrightarrow}$&
\shorttikz{\smallbasistwo \starD}
\end{tabular}\end{equation}
after which the sequence repeats.  Therefore, provided that the number of $F$ terms in the product is a multiple of three, the basis returns to the beginning after one final application of $R$.  
Similarly, the initial basis $\raisebox{.3em}{\scalebox{.5}{\shorttikz{\smallbasisone \starD}}}$ transforms as 
\begin{equation}
\begin{tabular}{l c l c l}
\shorttikz{\smallbasisone \starD}
&$\overset{F}{\longrightarrow}$&
\shorttikz{\smallbasistwo \starD}
&$\overset{R^{\alpha_1}}{\longrightarrow}$&
\shorttikz{\smallbasisone \starB} \\
&$\overset{F}{\longrightarrow}$&
\shorttikz{\smallbasistwo \starB}
&$\overset{R^{\alpha_2}}{\longrightarrow}$&
\shorttikz{\smallbasistwo \starC} \\
&$\overset{F}{\longrightarrow}$&
\shorttikz{\smallbasisone \starC}
&$\overset{R^{\alpha_3}}{\longrightarrow}$&
\shorttikz{\smallbasisone \starD}
\end{tabular}\end{equation}
after which the sequence repeats.  Therefore, provided that the number of $F$ terms in the product is $1 \mod 3$, the basis finishes as $\raisebox{.3em}{\scalebox{.5}{\shorttikz{\smallbasistwo \starD}}}$.  

We deduce that by setting $M_0 = F R^{-1} F R F$, so $m_j = 0 \pmod 3$, $M_j$ will converge doubly exponentially fast to $\fastmatrix{0&1\\1&0}$ up to phases---$\abs{\bra 0 M_j \ket 0} = 1/\tau^{2\cdot 5^j}$---and its implementation will finish in the correct basis position, therefore asymptotically implementing Eq.~\eqnref{e:addition}.  (We have chosen to start with $F R^{-1} F R F$ since $\bra 0 F R F R F \ket 0 = 1$.)  
By setting $N_0 = F$, $N_j$ will converge doubly exponentially fast to the identity up to phases---$\abs{\bra 1 N_j \ket 0} = 1/\sqrt \tau^{5^j}$---and its implementation will finish in the correct basis position for $j$ even, therefore asymptotically implementing Eq.~\eqnref{e:integration}.  In either case, the length of the braid sequence is $O(5^j)$; hence achieving error~$\epsilon$ requires $O(\log \frac 1 \epsilon)$ braids.  The braid sequences for $M_1$ and $N_1$ are given explicitly in \figref{f:braidsequences}.  

\begin{figure}
\centering
\begin{tabular}{c c}
\subfigure[]{\raisebox{.2cm}{\includegraphics[scale=.65]{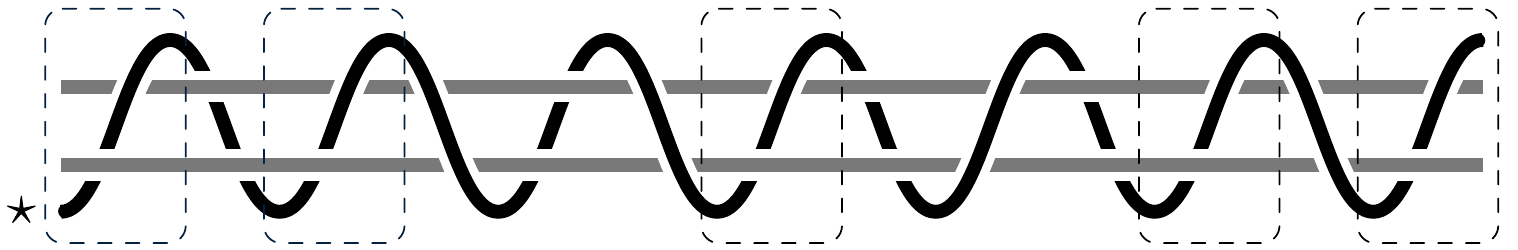}}} & 
\subfigure[]{\raisebox{.25cm}{\includegraphics[scale=.65]{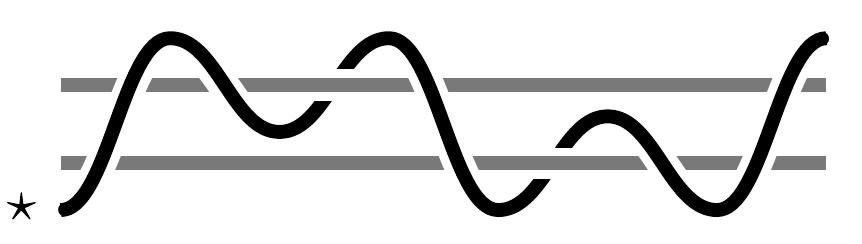}}}
\end{tabular}
\caption{(a) Braid sequence for $M_1$; the highlighted areas correspond to $M_0 = M_0^\dagger$.  Note that some adjacent braids cancel, allowing a modest simplification.  Two final braids, to implement either $R$ or~$R^{-1}$, will return the mobile quasiparticle to its initial position.  (b) Braid sequence for $N_1$.} \label{f:braidsequences}
\end{figure}

Therefore, to create a composite anyon except with probability $\epsilon$, it suffices to use $n = \Theta(\frac 1 p \log \frac 1 \epsilon)$ anyon pairs and $O(n^2 \log \frac n \epsilon) = O\big(\frac 1 {p^2} (\log \frac 1 \epsilon)^3\big)$ total braids, assuming $p = \Omega(\epsilon \log \frac 1 \epsilon)$.  This matches the asymptotics found in \secref{s:asymptoticsallmobile}.  In practice, though, it should require about twice as many anyon pairs as the earlier scheme, since not one but two nontrivial anyon pairs are needed for this catalyzed distillation scheme to succeed.

\section{Conclusion} \label{s:conclusion}

The cooling and initialization of a topological quantum computer present interesting challenges.  In this article, we have improved the efficiency of implementing certain gates useful for distilling composite anyons, in a model with limited quasiparticle mobility.  Instead of relying on the Solovay-Kitaev theorem, the gate compilation scheme is based on an explicit recursion, developed from a generalization of Grover's convergent search algorithm.  

One natural question is whether this systematic method of compiling gates can be extended to a larger gate set, yielding faster convergence more generically~\cite{HarrowRechtChuang01SolovayKitaev}.  We have only shown how to implement the identity gate (between different bases) and the NOT gate, up to phases.  Of course, as we use pure braids, an implementation of the CNOT gate up to phases immediately follows.  

A second question is whether this approach, or composite anyon distillation more generally, can be extended to more anyon models.  It appears likely that composite anyon distillation can be implemented for the $SO(3)_k$ theory, provided that $k+2$ is prime or $k = 7$.  A first step is to generalize the Solovay-Kitaev theorem to work in multiple sectors simultaneously.  It is known that if $\rho$ and $\sigma$ are unitary group representations of different dimensions, that are each separately dense, then the representation $\rho \oplus \sigma$ is also dense~\cite[Lemma~4.2]{AharonovArad06Jones}.  Using that the outer automorphism group of $SU(n)$ has at most two elements, the lemma can be extended to the case that $\rho$ and $\sigma$ have the same dimension.  The condition $k+2$ prime or $k = 7$ then arises from asking that the braid representations with different boundary conditions be dense and inequivalent.  Composite anyon distillation for $k = 7$ would extend the Turaev-Viro invariant BQP-completeness result of~\cite{AlagicJordanKoenigReichardt10tv}.  However, it is more practical to study the Ising model, $SU(2)_2$.   One case of interest is the model on a surface with nontrivial topology but with only an imperfect ability to distinguish~$0$ from $1$.  Another case to consider is the plane with noisy non-Clifford gates; can magic states distillation~\cite{BravyiKitaev04magic, Reichardt06magic, Reichardt04magic} be combined with composite anyon distillation?  

As mentioned in the introduction, one practical concern for schemes including ours is the possible presence of stray anyons trapped at surface defects.  Such anyons are particularly a problem for a measurement-based scheme where the interferometry regions change, or for any scheme that combines measurements with braiding.  The issue is that different operations may involve different subsets of anyons that should be treated collectively.  We will briefly sketch one possible solution.  Assume that we are granted a mobile anyon and one interferometer.  Assume moreover that the total charge in the interferometry region is nontrivial when the mobile anyon is outside, and is trivial when the anyon is moved inside.  (Provided that the mobile anyon is in fact nontrivial, the trivial measurement outcome can be \emph{forced}~\cite{BondersonFreedmanNayak08measurementonlytqc} by repeatedly measuring with the mobile anyon outside and inside the interferometer.)  Then there must be a pair of nontrivial anyons that fuse to the vacuum, one half of which is mobile and the other half trapped within the interferometer.  The mobile anyon can then be used to measure anyons in other regions of the device by braiding it around the other region, returning it to the interferometer and measuring the total charge.  With repetition, this measurement can reach arbitrarily high accuracy.  Static stray anyons do not pose a problem provided that the mobile anyon is moved along the same routes, away from any strays.  After initializing enough nontrivial anyon pairs in this way, universal quantum computation can be simulated following~\cite{SimonBonesteelFreedmanPetrovicHormozi05onemobile}.

\subsection*{Acknowledgements}

I gratefully acknowledge helpful discussions with Todd Brun, Aram Harrow, Robert K{\"o}nig and Vincent Nesme.  Research conducted at the Institute for Quantum Computing, University of Waterloo, supported by NSERC and ARO, and at the Kavli Institute for Theoretical Physics, supported by NSF grant PHY05-51164.

\bibliographystyle{alpha-eprint}
\bibliography{systematic}

\end{document}